\documentclass[final,1p]{elsarticle}
\usepackage{pslatex,bm} 
\usepackage{amsfonts,color,morefloats}
\usepackage{amssymb,amsmath,latexsym,amsthm}

\newcommand{\diff}{{\mathrm{diff}}} 
\newcommand{\wt}{{\mathrm{wt}}}

\newcommand{\tr}{{\mathrm{Tr}}}
\newcommand{\gf}{{\mathrm{GF}}}

\newcommand{\C}{{\mathcal{C}}}

\newcommand{\bc}{{\mathbf{c}}}

\newcommand{\bone}{{\mathbf{1}}}

\newcommand{\Segre}{ {{\mathrm{Segre}}} }
\newcommand{\Glynnone}{ {{\mathrm{Glynni}}} }
\newcommand{\Glynntwo}{ {{\mathrm{Glynnii}}} }
\newcommand{\Payne}{ {{\mathrm{Payne}}} }
\newcommand{\Trans}{ {{\mathrm{Trans}}} }

\newcommand{\Cherowitzo}{ {{\mathrm{Cherowitzo}}} } 
\newcommand{\Subiaco}{ {{\mathrm{Subiaco}}} }

\newtheorem{theorem}{Theorem}
\newtheorem{lemma}[theorem]{Lemma}

\newtheorem{corollary}[theorem]{Corollary}
\newtheorem{conj}[theorem]{Conjecture}

\newtheorem{example}[theorem]{Example}

\begin{document}

\begin{frontmatter}

\title{A Construction of Binary Linear Codes from Boolean Functions\tnoteref{fn1}}
\tnotetext[fn1]{
C. Ding's research was supported by
The Hong Kong Research Grants Council, under Proj. No. 16300415.
}


\author{Cunsheng Ding}
 \ead{cding@ust.hk}

 \address{Department of Computer Science and Engineering,
The Hong Kong University of Science and Technology, Clear Water Bay, Kowloon, Hong Kong}

\date{}

\begin{abstract}
Boolean functions have important applications in cryptography and coding theory. Two famous classes of binary codes derived 
from Boolean functions are the Reed-Muller codes and Kerdock codes. In the past two decades, a lot of progress on the 
study of applications of Boolean functions in coding theory has been made. Two generic constructions of binary linear 
codes with Boolean functions have been well investigated in the literature. The objective of this paper is twofold. The 
first is to provide a survey on recent results, and the other is to propose open problems on one of the two generic 
constructions of binary linear codes with Boolean functions. These open problems are expected to stimulate further research 
on binary linear codes from Boolean functions.  
\end{abstract}

\begin{keyword}
Almost bent functions, bent functions, difference sets, linear codes, semibent functions, o-polynomials. 
\end{keyword}

\end{frontmatter}

\section{Introduction}\label{sec-intro} 

Let $p$ be a prime and let $q=p^m$ for some positive integer $m$. 
An $[n,\, k,\, d]$ code $\C$ over $\gf(p)$ is a $k$-dimensional subspace of $\gf(p)^n$ with minimum 
(Hamming) distance $d$. 
Let $A_i$ denote the number of codewords with Hamming weight $i$ in a code
$\C$ of length $n$. The {\em weight enumerator} of $\C$ is defined by
$
1+A_1z+A_2z^2+ \cdots + A_nz^n.
$ 
The sequence $(1, A_1, A_2, \cdots, A_n)$ is called the \emph{weight distribution} of the code $\C$. 
A code $\C$ is said to be a $t$-weight code  if the number of nonzero
$A_i$ in the sequence $(A_1, A_2, \cdots, A_n)$ is equal to $t$.

Boolean functions are functions from $\gf(2^m)$ or $\gf(2)^m$ to $\gf(2)$. They are important building blocks 
for certain types of stream ciphers, and can also be employed to construct binary codes. Two famous families of 
binary codes are the Reed-Muller codes \cite{Reed,Muller} and Kerdock codes \cite{Carlet89,Carlet91,Kerdock}. 
In the literature two generic constructions of binary 
linear codes from Boolean functions have been well investigated. A lot of progress on the study of one of the two 
constructions has been made in the past decade. The objective of this paper is twofold. The first one is to provide 
a survey on recent development on this construction, and the other is to propose open problems on this generic 
constructions of binary linear codes with Boolean functions. 
These open problems are expected to stimulate further research on binary linear codes from Boolean functions.  

\section{Mathematical foundations} 

\subsection{Difference sets} 

For convenience later, we define the \emph{difference function\index{difference function}} of a subset 
$D$ of an abelian group $(A,\,+)$ as 
\begin{eqnarray}\label{eqn-DifferenceFunction}
\diff_D(x)=|D \cap (D+x)|, 
\end{eqnarray}
where $D+x=\{y+x: y \in D\}$. 

A subset $D$ of size $k$ in an abelian group $(A, \, +)$ with order $v$ is called 
a $(v, \,k, \,\lambda)$ \emph{difference set\index{difference set}}  in $(A,\,+)$ if the difference function 
$\diff_D(x)=\lambda$ for every nonzero $x \in A$. A difference set $D$ in $(A,\,+)$ is called \emph{cyclic} if 
the abelian group $A$ is so. 

Difference sets could be employed to construct linear codes in different ways. The reader is referred to 
\cite{DingDScodes,Ding15} for detailed information. Some of the codes presented in this survey paper 
are also defined by difference sets.

\subsection{Group characters in $\gf(q)$}

An {\em additive character} of $\gf(q)$ is a nonzero function $\chi$ 
from $\gf(q)$ to the set of nonzero complex numbers such that 
$\chi(x+y)=\chi(x) \chi(y)$ for any pair $(x, y) \in \gf(q)^2$. 
For each $b\in \gf(q)$, the function
\begin{eqnarray}\label{dfn-add}
\chi_b(c)=\epsilon_p^{\tr(bc)} \ \ \mbox{ for all }
c\in\gf(q) 
\end{eqnarray}
defines an additive character of $\gf(q)$, where and whereafter $\epsilon_p=e^{2\pi \sqrt{-1}/p}$ is 
a primitive complex $p$th root of unity and $\tr$ is the absolute trace function. When $b=0$,
$\chi_0(c)=1 \mbox{ for all } c\in\gf(q), 
$ 
and is called the {\em trivial additive character} of
$\gf(q)$. The character $\chi_1$ in (\ref{dfn-add}) is called the
{\em canonical additive character} of $\gf(q)$. 
It is known that every additive character of $\gf(q)$ can be 
written as $\chi_b(x)=\chi_1(bx)$ \cite[Theorem 5.7]{LN}. 

\subsection{Special types of polynomials over $\gf(q)$} 

It is well known that every function from $\gf(q)$ to $\gf(q)$ can be expressed as a polynomial over $\gf(q)$. 
A polynomial $f \in \gf(q)[x]$ is called a \emph{permutation polynomial\index{permutation polynomial}} 
if the associated polynomial function $f: a \mapsto f(a)$ from $\gf(q)$ to $\gf(q)$ is a permutation of 
$\gf(q)$. 

Dickson polynomials of the first kind over $\gf(q)$ are defined by  
\begin{eqnarray}\label{eqn-1stDP}
D_h(x, a)=\sum_{i=0}^{\lfloor \frac{h}{2} \rfloor} \frac{h}{h-i} \binom{h-i}{i} (-a)^i x^{h-2i}, 
\end{eqnarray} 
where $a \in \gf(q)$ and $h$ is called the {\em order} of the polynomial. Some of the linear codes that  
will be presented in this paper are defined by Dickson permutation polynomials of order 5 over $\gf(2^m)$.  

A polynomial $f \in \gf(q)[x]$ is said to be \emph{$e$-to-$1$} if the equation $f(x)=b$ over $\gf(q)$ 
has either $e$ solutions $x \in \gf(q)$ or no solution for every $b \in \gf(q)$, where $e \geq 1$ is an 
integer, and $e$ divides $q$.  By definition, permutation polynomials are 1-to-1. In this survey paper, we 
need $e$-to-$1$ polynomials over $\gf(2^m)$ for the construction of binary linear codes. 

\subsection{Boolean functions and their expressions} 

A function $f$ from $\gf(2^m)$ or $\gf(2)^m$ to $\gf(2)$ is called a {\em Boolean function}. 
A function $f$ from $\gf(2^m)$ to $\gf(2)$ is called {\em linear} if $f(x+y)=f(x)+f(y)$ for all $(x, y) \in \gf(2^m)^2$. 
A function $f$ from $\gf(2^m)$ to $\gf(2)$ is called {\em affine} if $f$ or $f-1$ is linear. 

The {\em Walsh transform} of $f: \gf(2^m) \to \gf(2)$ is defined by 
\begin{eqnarray}\label{eqn-WalshTransform2}
\hat{f}(w)=\sum_{x \in \gf(2^m)} (-1)^{f(x)+\tr(wx)} 
\end{eqnarray} 
where $w \in \gf(2^m)$. The {\em Walsh spectrum} of $f$ is the following multiset 
$$ 
\left\{\left\{ \hat{f}(w): w \in \gf(2^m) \right\}\right\}. 
$$ 

Let $f$ be a Boolean function from $\gf(2^m)$ to $\gf(2)$. The \emph{support} of $f$ is defined to be 
\begin{eqnarray}\label{eqn-Booleanfsupport}
D_f=\{x \in\gf(2^m) : f(x)=1\} \subseteq \gf(2^m). 
\end{eqnarray}
Clearly, $f \mapsto D_f$ is a one-to-one correspondence between the set of Boolean functions from  
$\gf(2^m)$ to $\gf(2)$  and the power set of $\gf(2^m)$.

\section{The first generic construction of linear codes from functions}\label{sec-1stgconst}

Let $f$ be any polynomial from $\gf(q)$ to $\gf(q)$, where $q=p^m$. A  code over $\gf(p)$ is defined by 
$$ 
\C(f)=\{ \bc=(\tr(af(x)+bx))_{x \in \gf(q)}: a \in \gf(q), \ b \in \gf(q)\},  
$$
where $\tr$ is the absolute trace function.
Its length is $q$, and its dimension is at most $2m$ and is equal to $2m$ in many cases. 
The dual of $\C(f)$ has dimension at least $q-2m$. 

Let $f$ be any polynomial from $\gf(q)$ to $\gf(q)$ such that $f(0)=0$. A  code over $\gf(p)$ is defined by 
$$ 
\C^*(f)=\{ \bc=(\tr(af(x)+bx))_{x \in \gf(q)^*}: a \in \gf(q), \ b \in \gf(q)\}. 
$$
Its length is $q-1$, and its dimension is at most $2m$ and is equal to $2m$ in many cases. 
The dual of $\C^*(f)$ has dimension at least $q-1-2m$. 

This is a generic construction of linear codes, which has a long history and its importance is supported by  
Delsarte's Theorem \cite{Dels75}. It gives a coding-theory characterisation of APN monomials, almost bent 
functions, and semibent functions (see, for examples, \cite{CCZ}, \cite{CCD00} and \cite{HX01}) when $q=2$. 
We will not deal with this construction in this paper.

\section{The second generic construction of linear codes from functions}\label{sec-2ndgconst} 

In this section, we present the second generic construction of linear codes over $\gf(p)$ with 
any subset $D$ of $\gf(p^m)$, and  introduce basic results about the linear codes. 
In Section \ref{sec-supportcodes}, we will consider specific families of binary linear codes 
from Boolean functions obtained with this generic construction.

\subsection{The description of the construction of linear codes}

Let $D=\{d_1, \,d_2, \,\ldots, \,d_n\} \subseteq \gf(q)$, where again $q=p^m$.
Recall that $\tr$ denotes the trace function from $\gf(q)$ onto $\gf(p)$ throughout 
this paper. We define a linear code of 
length $n$ over $\gf(p)$ by 
\begin{eqnarray}\label{eqn-maincode} 
\C_{D}=\{(\tr(xd_1), \tr(xd_2), \ldots, \tr(xd_n)): x \in \gf(q)\},   
\end{eqnarray}  
and call $D$ the \emph{defining set} of this code $\C_{D}$. By definition, the 
dimension of the code $\C_D$ is at most $m$.

This construction is generic in the sense that many classes of known codes could be produced 
by selecting the defining set $D \subseteq \gf(q)$ properly. This  construction technique was 
employed in \cite{Ding15}, \cite{DLN}, \cite{DN07}, \cite{DDing} and other papers for obtaining 
linear codes with a few weights. If the set $D$ is properly chosen, the code $\C_D$ may have 
good or optimal parameters. Otherwise, the code $\C_D$ could have bad parameters.

\subsection{The weights in the linear codes $\C_D$}

It is convenient to define for each $x \in \gf(q)$, 
\begin{eqnarray}\label{eqn-mcodeword}
\bc_{x}=(\tr(xd_1), \,\tr(xd_2), \,\ldots, \,\tr(xd_n)).  
\end{eqnarray} 
The Hamming weight $\wt(\bc_x)$ of $\bc_x$ is $n-N_x(0)$, where  
$$ 
N_x(0)=\left|\{1 \le i \le n: \tr(xd_i)=0\}\right| 
$$ 
for each $x \in \gf(q)$. 

It is easily seen that for any $D=\{d_1,\,d_2,\,\ldots, \,d_n\} \subseteq \gf(q)$
we have 
\begin{eqnarray}\label{eqn-hn3}  
pN_x(0) 
= \sum_{i=1}^n \sum_{y \in \gf(p)} e^{2\pi \sqrt{-1} y\tr(xd_i)/p} \nonumber 
= \sum_{i=1}^n \sum_{y \in \gf(p)} \chi_1(yxd_i) \nonumber 
= n + \sum_{y \in \gf(p)^*} \chi_1(yxD) 
\end{eqnarray} 
where $\chi_1$ is the canonical additive character of $\gf(q)$, $aD$ denotes the set 
$\{ad: d \in D\}$, and $\chi_1(S):=\sum_{x \in S} \chi_1(x)$ for any subset $S$ of $\gf(q)$.  
Hence, 
\begin{eqnarray}\label{eqn-weight}
\wt(\bc_x)=n-N_x(0)=\frac{(p-1)n-\sum_{y \in \gf(p)^*} \chi_1(yxD)}{p}. 
\end{eqnarray}

\subsection{Differences between the first and second generic constructions} 

The second generic construction of this section is different from the first generic construction of Section 
\ref{sec-1stgconst} in the following aspects: 
\begin{itemize}
\item While the length of the codes in the first generic construction in Section \ref{sec-1stgconst} is either 
         $q$ or $q-1$, that of the codes in the second generic construction could be any integer between 1 and $q$, depending on 
         the underlying defining set $D$. 
\item  While the dimension of the codes in the first construction in Section \ref{sec-1stgconst} is usually $2m$, 
           that of the codes in the second construction is usually $m$ and is at most $m$. 
\end{itemize}

\section{Binary codes from the preimage $f^{-1}(b)$ of Boolean functions $f$}\label{sec-supportcodes}

Let $f$ be a function from $\gf(2^m)$ to $\gf(2)$, and let $D$ be any subset of the preimage $f^{-1}(b)$ for any 
$b \in \gf(2)$. In general, it is very hard to determine the parameters of the code $\C_{D}$. 
Recall the support $D_f$ of $f$ defined in (\ref{eqn-Booleanfsupport}).  Let $n_f=|D_f|$.  
In this section, we deal with the binary code $\C_{D_f}$ with length $n_f$ and dimension at most $m$, and will 
focus on the weight distribution of the code $\C_{D_f}$ for several classes of Boolean functions $f$.

The following theorem plays a major role in this section whose proof can be found in \cite{Ding15}. 

\begin{theorem}\label{thm-BooleanCodes}
Let $f$ be a function from $\gf(2^m)$ to $\gf(2)$, and let $D_f$ be the support of $f$. If $2n_f + \hat{f}(w) \neq 0$ 
for all $w \in \gf(2^m)^*$, then $\C_{D_f}$ is a binary linear code with 
length $n_f$ and dimension $m$, and its weight distribution is given by the following multiset: 
\begin{eqnarray}\label{eqn-WTBcodes}
\left\{\left\{ \frac{2n_f+\hat{f}(w)}{4}: w \in \gf(2^m)^*\right\}\right\} \cup \left\{\left\{ 0 \right\}\right\}. 
\end{eqnarray} 
\end{theorem}

Theorem \ref{thm-BooleanCodes} establishes a connection between the set of Boolean functions $f$ such that 
$2n_f + \hat{f}(w) \neq 0$ for all $w \in \gf(2^m)^*$ and a class of binary linear codes. 
The determination of the weight distribution of the binary linear code $\C_{D_f}$ is equivalent to that of the Walsh 
spectrum of the Boolean function $f$ satisfying $2n_f + \hat{f}(w) \neq 0$ for all $w \in \gf(2^m)^*$. 
When the Boolean function $f$ is selected properly, the code $\C_{D_f}$ has 
only a few weights and may have good parameters. We will demonstrate this in the remainder of this section.  

We point out that Theorem \ref{thm-BooleanCodes} can be generalized into the following whose proof is the same 
as that of Theorem \ref{thm-BooleanCodes}. 

\begin{theorem}\label{thm-BooleanCodesG}
Let $f$ be a function from $\gf(2^m)$ to $\gf(2)$, and let $D_f$ be the support of $f$. Let $e_w$ denote the 
multiplicity of the element $\frac{2n_f+\hat{f}(w)}{4}$ and $e$ the multiplicity of 0  in the following multiset 
of (\ref{eqn-WTBcodes}). 
Then $\C_{D_f}$ is a binary linear code with length $n_f$ and dimension $m-\log_2 e$, and 
the weight distribution of the code is given by  
\begin{eqnarray*}
\frac{2n_f+\hat{f}(w)}{4} \mbox{ with frequency } \frac{e_w}{e} 
\end{eqnarray*}
for all $\frac{2n_f+\hat{f}(w)}{4}$ in the multiset of (\ref{eqn-WTBcodes}). 
\end{theorem}

\subsection{Linear codes from bent functions} 

A function from $\gf(2^m)$ to $\gf(2)$ is called \emph{bent\index{bent}} if $|\hat{f}(w)|=
2^{m/2}$ for every $w \in \gf(2^m)$. Bent functions exist only for even $m$ \cite{Rothaus76}.  

It is well known that 
a function $f$ from $\gf(2^m)$ to $\gf(2)$ is bent if and only if $D_f$ is 
a difference set in  $(\gf(2^m),\,+)$ with the following parameters 
\begin{eqnarray}\label{eqn-MenonHadamardPara}
(2^m, \, 2^{m-1} \pm 2^{(m-2)/2}, \, 2^{m-2} \pm 2^{(m-2)/2}).  
\end{eqnarray} 

Let $f$ be bent. Then by definition $\hat{f}(0)=\pm 2^{m/2}$. It then follows that 
\begin{eqnarray}\label{eqn-bentfuncsupportsize}
n_f=|D_f|=2^{m-1} \pm 2^{(m-2)/2}
\end{eqnarray}

\begin{table}[ht]
\begin{center} 
\caption{The weight distribution of the codes of Corollary \ref{thm-bentcodes}}\label{tab-bentfcode}
\begin{tabular}{cc} \hline
Weight $w$ &  Multiplicity $A_w$  \\ \hline  
$0$          &  $1$ \\  
$\frac{n_f}{2}-2^{\frac{m-4}{2}}$          &  $\frac{2^m-1-n_f2^{-\frac{m-2}{2}}}{2}$ \\    
$\frac{n_f}{2}+2^{\frac{m-4}{2}}$          &  $\frac{2^m-1+n_f2^{-\frac{m-2}{2}}}{2}$ \\ \hline  
\end{tabular}
\end{center} 
\end{table}

As a corollary of Theorem \ref{thm-BooleanCodes}, we have the following \cite{Ding15}. 

\begin{corollary}\label{thm-bentcodes}
Let $f$ be a Boolean function from $\gf(2^m)$ to $\gf(2)$ with $f(0)=0$, where $m \geq 4$ and is even. 
Then $\C_{D_f}$ is an $[n_f, \,m, \,(n_f-2^{(m-2)/2})/2]$ 
two-weight binary code with the weight distribution in Table \ref{tab-bentfcode}, where $n_f$ is defined in (\ref{eqn-bentfuncsupportsize}), if and only if $f$ is bent. 
\end{corollary} 

There are many constructions of bent functions and thus Hadamard difference sets. We refer the reader to 
\cite{BCHKM}, \cite{Mesnager142}, \cite{Mesnager143},  the book chapter \cite{Carlet} and the references 
therein for details. Any bent function can be plugged into Corollary \ref{thm-bentcodes} to obtain a two-weight 
binary linear code. 

The construction of binary codes with bent functions above can be generalized as follows. 

\begin{theorem}\label{thm-part1222}
Let $D$ be a $(2^m, n, \, \lambda)$ difference sets in $(\gf(2^m), +)$. 
Then $\C_D$ is a two-weight binary code with parameters $[n,\, m]$ 
and weight enumerator 
\begin{eqnarray*}
1+ \frac{(2^m-1)\sqrt{n-\lambda} - n}{2\sqrt{n-\lambda}} z^{\frac{n-\sqrt{n-\lambda}}{2} } +  
\frac{(2^m-1)\sqrt{n-\lambda} + n}{2\sqrt{n-\lambda}} z^{\frac{n + \sqrt{n-\lambda}}{2}}.   
\end{eqnarray*} 
\end{theorem}

\subsection{Linear codes from semibent functions} 

Let $m$ be odd. Then there is no bent Boolean function on $\gf(2^m)$.  
A function $f$ from $\gf(2^m)$ to $\gf(2)$ is called \emph{semibent} if $\hat{f}(w) \in \{0, \,
\pm 2^{(m+1)/2}\}$ for every $w \in \gf(2^m)$. 

Let $f$ be a semibent function from $\gf(2^m)$ to $\gf(2)$. It then follows from the definition of semibent functions 
that  
\begin{eqnarray}\label{eqn-semibf}
n_f=|D_f| = \left\{ \begin{array}{ll}
                           2^{m-1}-2^{(m-1)/2} & \mbox{ if } \hat{f}(0)=2^{(m+1)/2}, \\
                           2^{m-1}+2^{(m-1)/2} & \mbox{ if } \hat{f}(0)=-2^{(m+1)/2}, \\       
                           2^{m-1}  & \mbox{ if } \hat{f}(0)=0.                                                
\end{array}
\right. 
\end{eqnarray}

\begin{table}[ht]
\begin{center} 
\caption{The weight distribution of the codes of Corollary \ref{thm-semibentcodes}}\label{tab-semibentfcode}
\begin{tabular}{cc} \hline
Weight $w$ &  Multiplicity $A_w$  \\ \hline  
$0$          &  $1$ \\  
$\frac{n_f-2^{(m-1)/2}}{2}$ & $n_f(2^m-n_f)2^{-m} - n_f2^{-(m+1)/2}$ \\ 
$\frac{n_f}{2}$ & $2^m-1-n_f(2^m-n_f)2^{-(m-1)}$ \\ 
$\frac{n_f+2^{(m-1)/2}}{2}$ & $n_f(2^m-n_f)2^{-m} + n_f2^{-(m+1)/2}$ \\ \hline 
\end{tabular}
\end{center} 
\end{table} 

As a corollary of Theorem \ref{thm-BooleanCodes}, we have the following \cite{Ding15}. 

\begin{corollary}\label{thm-semibentcodes}
Let $f$ be a Boolean function from $\gf(2^m)$ to $\gf(2)$ with $f(0)=0$, where $m$ is odd. 
Then $\C_{D_f}$ is an $[n_f, \,m, \,(n_f-2^{(m-1)/2})/2]$ 
three-weight binary code with the weight distribution in Table \ref{tab-semibentfcode}, where $n_f$ is defined in (\ref{eqn-semibf}), if and only if $f$ is semibent.  
\end{corollary} 

There are a lot of constructions of semibent functions from $\gf(2^m)$ to $\gf(2)$. We refer the reader to 
\cite{CM12,CM14, DQFL,Mesnager11,Mesnager13,Mesnager14} for detailed constructions. All semibent functions can be plugged into Corollary \ref{thm-semibentcodes} to obtain three-weight binary linear codes. 

\subsection{Linear codes from almost bent functions}

For any function $g$ from $\gf(2^m)$ to $\gf(2^m)$, we define 
$$ 
\lambda_g(a, b) = \sum_{x \in \gf(2^m)} (-1)^{\tr(ag(x)+bx)}, \ a,\, b \in \gf(2^m).     
$$ 
A function $g$ from $\gf(2^m)$ to $\gf(2^m)$ is called {\em almost bent (AB)} if 
$\lambda_g(a, b) = 0, \mbox{ or } \pm 2^{(m+1)/2}$ for every pair $(a, b)$ with $a \neq 0$. 
By definition, almost bent functions over $\gf(2^m)$ exist only for odd $m$.  
Specific almost bent functions are available in \cite{BCP,Carlet}.

By definition,  $\lambda_g(1, 0) \in \{0, \,\pm 2^{(m+1)/2}\}$ for any almost bent function $g$ on $\gf(2^m)$. It is straightforward 
to deduce the following lemma. 

\begin{lemma}\label{lem-absize}
For any almost bent function $g$ from $\gf(2^m)$ to $\gf(2^m)$, define $f=\tr(g)$. Then we have  
\begin{eqnarray}\label{eqn-newnf2}
n_{f}=|D_{\tr(g)}| &=& \left\{ \begin{array}{ll}
                           2^{m-1}+2^{(m-1)/2} & \mbox{ if } \lambda_g(1, 0)=-2^{(m+1)/2}, \\
                           2^{m-1}-2^{(m-1)/2} & \mbox{ if } \lambda_g(1, 0)=2^{(m+1)/2}, \\       
                           2^{m-1}  & \mbox{ if } \lambda_g(1, 0)=0.                                                
\end{array}
\right. 
\end{eqnarray} 
\end{lemma}

As a corollary of Theorem \ref{thm-BooleanCodes}, we have the following \cite{Ding15}. 

\begin{corollary}\label{thm-abcodes}
Let $g$ be an almost bent function from $\gf(2^m)$ to $\gf(2^m)$ with $\tr(g(0))=0$, where $m$ is odd. Define $f=\tr(g)$. 
Then $\C_{D_{f}}$ is an $[n_f, \,m, \,(n_f-2^{(m-1)/2})/2]$ 
three-weight binary code with the weight distribution in Table \ref{tab-semibentfcode}, where $n_f$ is given in (\ref{eqn-newnf2}). 
\end{corollary} 

The following is a list of almost bent functions $g(x)=x^d$ on $\gf(2^m)$ for odd $m$: 
\begin{enumerate}
\item $d=2^h+1$, where $\gcd(m, h)=1$ is odd \cite{Gold68}. 
\item $d=2^{2h}-2^h+1$, where $h \geq 2$ and $\gcd(m, h)=1$ is odd \cite{Kasami71}.  
\item $d=2^{(m-1)/2}+3$ \cite{Kasami71}. 
\item $d=2^{(m-1)/2}+2^{(m-1)/4}-1$, where $m \equiv 1 \pmod{4}$ \cite{HX01,Hou}. 
\item $d=2^{(m-1)/2}+2^{(3m-1)/4}-1$, where $m \equiv 3 \pmod{4}$ \cite{HX01,Hou}. 
\end{enumerate}
This list of almost bent monomials $g(x)$ are permutation polynomials on $\gf(2^m)$. Hence, 
the length of the code $\C_{D_{f}}$ is equal to $2^{m-1}$, and the weight distribution of the 
code is given in Table  \ref{tab-semibentfcode6}.

\subsection{Linear codes from quadratic Boolean functions} 

Let 
\begin{eqnarray}\label{eqn-QBFs}
f(x)=\tr_{2^m/2} \left(  \sum_{i=0}^{\lfloor m/2 \rfloor} f_i x^{2^i +1} \right) 
\end{eqnarray}
be a quadratic Boolean function from $\gf(2^m)$ to $\gf(2)$, where $f_i \in \gf(2^m)$. The rank of 
$f$, denoted by $r_f$, is defined to be the codimension of the $\gf(2)$-vector space 
$$ 
V_f=\{x \in \gf(2^m): f(x+z)-f(x)-f(z)=0 \ \forall \ z \in \gf(2^m)\}. 
$$
The Walsh spectrum of $f$ is known \cite{CCK} and given in Table \ref{tab-WalshBBFs}.  

\begin{table}[ht]
\begin{center} 
\caption{The Walsh spectrum of quadratic Boolean functions}\label{tab-WalshBBFs}
\begin{tabular}{cc} \hline
$\hat{f}(w)$ &  the number of $w$'s  \\ \hline  
$0$          &  $2^m-2^{r_f}$ \\  
$2^{m-r_f/2}$          &  $2^{r_f-1}+2^{(r_f-2)/2}$ \\ 
$-2^{m-r_f/2}$          &  $2^{r_f-1}-2^{(r_f-2)/2}$ \\ \hline 
\end{tabular}
\end{center} 
\end{table} 

Let $D_f$ be the support of $f$. By definition, we have 
\begin{eqnarray}\label{eqn-QBFcodeL}
n_f=|D_f|=2^{m-1}-\frac{\hat{f}(0)}{2} 
=\left\{ 
\begin{array}{ll}
2^{m-1}   & \mbox{ if } \hat{f}(0)=0, \\
2^{m-1}-2^{m-1-r_f/2}   & \mbox{ if } \hat{f}(0)=2^{m-r_f/2}, \\
2^{m-1}+2^{m-1-r_f/2}   & \mbox{ if } \hat{f}(0)=-2^{m-r_f/2}. 
\end{array}
\right. 
\end{eqnarray}

The following theorem then follows from Theorem \ref{thm-BooleanCodes} and Table \ref{tab-WalshBBFs}. 

\begin{theorem}\label{thm-CodeQBFs}
Let $f$ be a quadratic Boolean function of the form in (\ref{eqn-QBFs}) such that $r_f > 2$. 
Then $\C_{D_f}$ is a binary code with length $n_f$ given in (\ref{eqn-QBFcodeL}), 
dimension $m$, and the weight distribution in Table \ref{tab-WEqbfs}, where 
\begin{eqnarray}
(\epsilon_1, \epsilon_1, \epsilon_3)=
\left\{ 
\begin{array}{ll}
(1,0,0)   & \mbox{ if } \hat{f}(0)=0, \\
(0,1,0)   & \mbox{ if } \hat{f}(0)=2^{m-1-r_f/2}, \\
(0,0,1)   & \mbox{ if } \hat{f}(0)=-2^{m-1-r_f/2}. 
\end{array}
\right. 
\end{eqnarray} 
\end{theorem}

\begin{table}[ht]
\begin{center} 
\caption{The weight distribution of the code $\C_{D_{f}}$ in Theorem \ref{thm-CodeQBFs}}\label{tab-WEqbfs}
\begin{tabular}{cc} \hline
Weight $w$ &  $A_w$  \\ \hline  
$0$          &  $1$ \\ 
$\frac{n_f}{2}$          &  $2^m-2^{r_f}-\epsilon_1$ \\  
$\frac{n_f+2^{m-1-r_f/2}}{2}$          &  $2^{r_f-1}+2^{(r_f-2)/2}-\epsilon_2$ \\ 
$\frac{n_f-2^{m-1-r_f/2}}{2}$          &  $2^{r_f-1}-2^{(r_f-2)/2}-\epsilon_3$ \\ \hline 
\end{tabular}
\end{center} 
\end{table} 

Note that the code $\C_{D_f}$ in Theorem \ref{thm-CodeQBFs} defined by any quadratic Boolean function $f$ is 
different from any subcode of the second-order Reed-muller code, due to the difference in their lengths. The weight 
distributions of the two codes are also different.   

\subsection{Some binary codes $\C_{D_f}$ with three weights}

\begin{table}[ht]
\begin{center} 
\caption{Boolean functions with three-valued Walsh spectrum}\label{tab-3valuespectrum} 
\begin{tabular}{cc} \hline
$\hat{f}(w)$ &  the number of $w$'s  \\ \hline  
$0$          &  $2^m-2^{m-e}$ \\  
$2^{(m+e)/2}$          &  $2^{m-e-1}+2^{(m-e-2)/2}$ \\  
$-2^{(m+e)/2}$          &  $2^{m-e-1}-2^{(m-e-2)/2}$ \\ \hline 
\end{tabular}
\end{center} 
\end{table} 

\begin{theorem}
Let $m \geq 4$ be even. Then the code $\C_{D_f}$ has parameters $[2^{m-1},\, m,\, 2^{m-2}-2^{(m-2)/2}]$ 
and the weight distribution of Table \ref{tab-3wtc1}, where $e=2$, for $f(x)=\tr(x^d)$ for the following $d$: 
\begin{enumerate}
\item $d=2^h+1$, where $\gcd(m, h)$ is odd and $1 \leq h \leq m/2$ \cite{Gold68}. 
\item $d=2^{2h}-2^h+1$, where $\gcd(m, h)$ is odd and $1 \leq h \leq m/2$ \cite{Kasami71}.  
\item $d=2^{m/2}+2^{(m+2)/4}+1$, where $m \equiv 2 \pmod{4}$ \cite{CusickDob}.  
\item $d=2^{(m+2)/2}+3$, where $m \equiv 2 \pmod{4}$ \cite{CusickDob}. 
\end{enumerate}
\end{theorem}

\begin{proof}
It can be verified that $\gcd(d, 2^m-1)=1$ for all the $d$ listed above. Hence $n_f=|D_f|=2^{m-1}$.  
The Walsh spectrum of the functions $f$ above is given in Table \ref{tab-3valuespectrum} according 
to the references given in this theorem. The desired conclusions on the parameters and the weight 
distribution of the code $\C_{D_f}$ then follow from Theorem \ref{thm-BooleanCodes}. 
\end{proof}

\begin{table}[ht]
\begin{center} 
\caption{The weight distribution of some three-weight codes}\label{tab-3wtc1}
\begin{tabular}{cc} \hline
Weight $w$ &  Multiplicity $A_w$  \\ \hline  
$0$     &  $1$ \\  
$2^{m-2}$          &  $2^m-2^{m-e}-1$ \\  
$2^{m-2}+2^{(m+e-4)/2}$          &  $2^{m-e-1}+2^{(m-e-2)/2}$ \\ 
$2^{m-2}-2^{(m+e-4)/2}$          &  $2^{m-e-1}-2^{(m-e-2)/2}$ \\ \hline 
\end{tabular}
\end{center} 
\end{table} 

\subsection{Binary codes $\C_{D_f}$ with four weights}

\begin{table}[ht]
\begin{center} 
\caption{Boolean functions with four-valued Walsh spectrum: Case I}\label{tab-4valuespectrum1}
\begin{tabular}{cc} \hline
$\hat{f}(w)$ &  the number of $w$'s  \\ \hline  
$-2^{m/2}$          &  $(2^{m}-2^{m/2})/3$ \\ 
$0$          &  $2^{m-1}-2^{(m-2)/2}$ \\  
$2^{m/2}$          &  $2^{m/2}$ \\ 
$2^{(m+2)/2}$          &  $(2^{m-1}-2^{(m-2)/2})/3$ \\ \hline 
\end{tabular}
\end{center} 
\end{table}

\begin{table}[ht]
\begin{center} 
\caption{Boolean functions with four-valued Walsh spectrum: Case 2}\label{tab-4valuespectrum2}
\begin{tabular}{cc} \hline
$\hat{f}(w)$ &  the number of $w$'s  \\ \hline  
$-2^{m/2}$          &  $2^{m-1}-2^{(3m-4)/4}$ \\  
$0$          &  $2^{3m/4}-2^{m/4}$ \\  
$2^{m/2}$          &  $2^{m-1}-2^{(3m-4)/4}$ \\ 
$2^{3m/4}$          &  $2^{m/4}$ \\ \hline 
\end{tabular}
\end{center} 
\end{table} 

\begin{table}[ht]
\begin{center} 
\caption{Boolean functions with four-valued Walsh spectrum: Case 3}\label{tab-4valuespectrum3}
\begin{tabular}{cc} \hline
$\hat{f}(w)$ &  the number of $w$'s  \\ \hline  
$-2^{m/2}$          &  $(2^{m}-2^{m/2}-2)/3$ \\  
$0$          &  $2^{m-1}-2^{(m-2)/2}+2$ \\  
$2^{m/2}$          &  $2^{m/2}-2$ \\  
$2^{(m+2)/2}$          &  $(2^{m-1}-2^{(m-2)/2}+2)/3$ \\ \hline 
\end{tabular}
\end{center} 
\end{table}

The code $\C_{D_f}$ has four weights and its weight distribution is known when $f(x)=\tr(x^d)$ and $d$ is given 
in the following list.   
\begin{itemize}
\item When $d=2^{(m+2)/2}-1$ and $m \equiv 0 \pmod{4}$, 
          the  code $\C_{D_f}$ has length $2^{m-1}$ and dimension $m$, 
          and the weight distribution of $\C_{D_f}$ is deduced from Theorem \ref{thm-BooleanCodes} and 
          Table \ref{tab-4valuespectrum1}, where $h=1$ \cite{Niho}. 

\item When $d=2^{(m+2)/2}-1$ and $m \equiv 2 \pmod{4}$, 
          the  code $\C_{D_f}$ has length $2^{m-1}-2^{m/2}$ and dimension $m$,  
          and the weight distribution of $\C_{D_f}$ is deduced from Theorem \ref{thm-BooleanCodes} and 
          Table \ref{tab-4valuespectrum3} \cite{Niho}.  Note that in this case, 
          $\gcd(d, 2^m-1)=3$.            
          
\item When $d=(2^{m/2}+1)(2^{m/4}-1)+2$ and $m \equiv 0 \pmod{4}$, 
          the  code $\C_{D_f}$ has length $2^{m-1}$ and dimension $m$, 
          and the weight distribution of $\C_{D_f}$ is deduced from Theorem \ref{thm-BooleanCodes} and 
          Table \ref{tab-4valuespectrum2} \cite{Niho}.  

\item When $d=\frac{2^{(m+2)h/2}-1}{2^h-1}$ and $m \equiv 0 \pmod{4}$, where $1 \leq h < m$ and $\gcd(h, m)=1$, 
          the  code $\C_{D_f}$ has length $2^{m-1}$ and dimension $m$, 
          and the weight distribution of $\C_{D_f}$ is deduced from Theorem \ref{thm-BooleanCodes} and 
          Table \ref{tab-4valuespectrum1}, where $h=1$ \cite{Dobb98}. 
              
\item When $d=\frac{2^{(m+2)h/2}-1}{2^h-1}$ and $m \equiv 2 \pmod{4}$, where  $1 \leq h < m$ and $\gcd(h, m)=1$,     
          the  code $\C_{D_f}$ has length $2^{m-1}-2^{m/2}$ and dimension $m$,  
          and the weight distribution of $\C_{D_f}$ is deduced from Theorem \ref{thm-BooleanCodes} and 
          Table \ref{tab-4valuespectrum3} \cite{Niho}.  Note that in this case, 
          $\gcd(d, 2^m-1)=3$.                        
    
\item When $d=\frac{2^m+2^{h+1}-2^{m/2+1}-1}{2^h-1}$, where $2h$ divides $m/2$ and  
          $m \equiv 0 \pmod{4}$, 
          the  code $\C_{D_f}$ has length $2^{m-1}$ and dimension $m$, 
          and the weight distribution of $\C_{D_f}$ is deduced from Theorem \ref{thm-BooleanCodes} and 
          Table \ref{tab-4valuespectrum1} \cite{HR05}.  
             
\item When $d=(2^{m/2}-1)s+1$ with $s=2^h(2^h \pm 1)^{-1} \pmod{2^{m/2}+1}$,  where $e_2(h) < e_2(m/2)$ 
          and $e_2(h)$ denotes the highest power of $2$ dividing $h$, the parameters and the weight distribution of the 
          code $\C_{D_f}$ can be deduced from Theorem \ref{thm-BooleanCodes} and  the results in \cite{DFHR}.      
           
\item Let $d$ be any integer such that $1 \leq d \leq 2^m-2$ and $d(2^\ell +1) \equiv 2^h \pmod{2^m-1}$ for some 
         positive integers $\ell$ and $h$. Then the parameters and the weight distribution of the code $\C_{D_f}$ can be 
         deduced from Theorem \ref{thm-BooleanCodes} and the results in \cite{HHKZLJ}.                
\end{itemize}

All these cases of $d$ above are derived from the cross-correlation of a binary maximum-length 
sequence with its $d$-decimation version.

\subsection{Other binary codes $\C_{D_f}$ with at most five weights} 

The code $\C_{D_f}$ has at most five weights for the following $f$: 
\begin{itemize}
\item When $f(x)=\tr(x^{2^{m/2}+3})$, where $m \geq 6$ and is even,  $\C_{D_f}$ is a five-weight code with 
          length $2^{m-1}$ and dimension $m$, and its weight distribution can be derived from \cite{Helleseth76}. 

\item When $f(x)=\tr(ax^{(2^m-1)/3})$ with $\tr_{4}^{2^m}(a) \ne 0$, where $m$ is even, $\C_{D_f}$ is a 
          two-weight code with length $(2^{m+2}-4)/6$ and dimension $m$, and its weight distribution can be derived from 
          \cite{LiYue}. 


\item When $f(x)=\tr(\lambda x^{2^{m/2}+1}) +\tr(x)\tr(\mu x^{2^{m/2}-1})$, 
          where $m$ is even, $\mu \in \gf(2^{m/2})^*$, and $\lambda \in \gf(2^m)$ 
          with $\lambda + \lambda^{2^m}=1$, $\C_{D_f}$ is a five-weight code \cite{CaoHu15}. 
\item When $f(x)=(1+\tr(x))\tr(\lambda x^{2^{m/2}+1}) +\tr(x)\tr(\mu x^{2^{m/2}-1})$, 
          where $m$ is even, $\mu \in \gf(2^{\frac{m}{2}})^*$, and $\lambda \in \gf(2^m)$ 
          with $\lambda + \lambda^{2^m}=1$, $\C_{D_f}$ is a five-weight code \cite{CaoHu15}.           
\end{itemize}
Some Boolean functions $f$ documented in \cite{Quetal} gives also binary linear code $\C_{D_f}$ with five 
weights.

\subsection{A class of two-weight codes from the preimage of a type of Boolean functions}

Let $m$ be a positive integer and let $r$ be a prime such that $2$ is a primitive root modulo 
$r^m$. Let $q=2^{\phi(r^m)}$, where $\phi$ is the Euler function. 
Define 
\begin{eqnarray}\label{eqn-wdxD}
D=\left\{ x \in \gf(q)^*: \tr\left( x^{\frac{q-1}{r^m}} \right) =0 \right\}.  
\end{eqnarray}

The following theorem was proved in \cite{WDX}. 

\begin{theorem}\label{thm-WDX}
Let $r^m \geq 9$ and let $D$ be defined in (\ref{eqn-wdxD}). Then the set $\C_D$ of (\ref{eqn-maincode}) 
is a binary code with length $(q-1)(r^m-r+1)/r^m$, dimension $(r-1)r^{m-1}$ and the weight distribution 
in Table \ref{tal:cyclotomy40}.
\end{theorem}

 \begin{table}
\centering
\caption{The weight distribution of the codes of Theorem \ref{thm-WDX}}\label{tal:cyclotomy40}
\begin{tabular}{ll}
\hline
\textrm{weight} $w$ \qquad& \textrm{Multiplicity} $A_{w}$   \\
\hline
0 \qquad&   1  \\
$\frac{q-\sqrt{q}}{4}+\frac{q+\sqrt{q}}{4r^{m}}(r^m-2r+2)$ \qquad&  $\frac{(q-1)(r^{m}-r+1)}{r^{m}}$  \\
$\frac{q+\sqrt{q}}{4}+\frac{q+\sqrt{q}}{4r^{m}}(r^m-2r+2)$  \qquad& $\frac{(q-1)(r-1)}{r^{m}}$  \\
\hline
\end{tabular}
\end{table}

\subsection{Binary codes from Boolean functions whose supports are relative difference sets}

Let $(A, +)$ be an abelian group of order $m\ell$ and $(N, +)$ a subgroup of $A$ of order $\ell$. 
A $n$-subset $D$ of $A$ is called an $(m, \ell, n, \lambda)$ \textit{relative difference set}, if the 
multiset $\{\{ d_1-d_2: d_1, \, d_2 \in D,\ d_1 \neq d_2 \}\}$ does not contain all elements in 
$N$, but every element in $A \setminus N$ exactly $\lambda$ times.    

It well known in combinatorics that 
$ 
|\chi(D)|^2 \in \{n, n-\lambda \ell\} 
$
for any nontrivial group character $\chi$. Hence, any relative difference set $D$ in $(\gf(2^m), +)$ defines 
a binary code $\C_D$ with at most the following four weights: 
$$
\frac{n \pm \sqrt{n}}{2}, \ \frac{n \pm \sqrt{n-\lambda \ell}}{2}. 
$$ 
Obviously, $D$ is the support of a Boolean function on $\gf(2^m)$.

\section{Binary codes from the images of certain functions on $\gf(2^m)$} 

Let $f(x)$ be a function from $\gf(2^m)$ to $\gf(2^m)$. We define 
$$ 
D(f) =\{f(x): x \in \gf(2^m)\} \mbox{ and }  D(f)^* =\{f(x): x \in \gf(2^m)\} \setminus \{0\}.  
$$
In this section, we consider the code $\C_{D(f)}$. In general, it is difficult to determine the length $n_f:=|D(f)|$ of this code, 
not to mention its weight distribution. However, in certain special cases, the parameters and the weight distribution of $\C_{D(f)}$ 
can be settled. 

If $0 \not\in D(f)$, then the two codes $\C_{D(f)}$ and $\C_{D(f)^*}$ are the same. Otherwise, the length of 
the code $\C_{D(f)^*}$ is one less than that of the code $\C_{D(f)}$, but the two codes have the same weight 
distribution. Thus, we will not give information about the codes $\C_{D(f)^*}$ in this section. 

Let $D$ be any subset of $\gf(2^m)$. The {\em characteristic function}, denoted by $f_D(x)$, of $D$ is defined by 
$$ 
f_D(x)= \left\{ \begin{array}{ll}
1 & \mbox{ if } x \in D, \\
0 & \mbox{ otherwise.} 
\end{array}
\right. 
$$ 
Hence, the Boolean function $f_D$ has support $D$. Thus, the code $\C_{D(f)}$ is in fact defined by the support 
of the characteristic function (Boolean function) of the set $D(f)$. Therefore, the construction method of this section 
is actually equivalent to that of Section \ref{sec-supportcodes}.

\subsection{The codes $\C_{D(f)}$ from o-polynomials on $\gf(2^m)$}\label{sec-opolycodes}  

A permutation polynomial $f$ on $\gf(2^m)$ is called an \emph{o-polynomial\index{o-polynomial}} if 
$f(0)=0$, and for each $s \in \gf(2^m)$, 
\begin{eqnarray}
f_s(x)=(f(x+s)+f(s))x^{2^m-2} 
\end{eqnarray} 
is also a permutation polynomial. O-polynomial can be used to construct hyperovals in finite geometry. 

In the original definition of o-polynomials, it is required that $f(1)=1$. However, this is not essential, as 
one can always normalise $f(x)$ by using $f(1)^{-1} f(x)$ due to that $f(1) \neq 0$. 

In this section, we consider binary codes  $\C_{D(f)}$, where $f$ is defined by an o-polynomial in some way. 

\subsubsection{O-polynomials and their binary codes $\C_{D(f_u)}$} 

For any permutation polynomial $f(x)$ over $\gf(2^m)$, we define $\overline{f}(x)=xf(x^{2^m-2})$, and use $f^{-1}$ to denote the compositional inverse of $f$, i.e., $f^{-1}(f(x))=x$ for all $x \in \gf(2^m)$. 

The following two theorems introduce basic properties of o-polynomials whose proofs can be 
found in references about hyperovals. . 

\begin{theorem}\label{thm-basicproperty}
Let $f$ be an o-polynomial on $\gf(2^m)$. Then the following statements hold: 
\begin{enumerate}
\item $f^{-1}$ is also an o-polynomial; 
\item $f(x^{2^{j}})^{2^{m-j}}$ is also an o-polynomial for any $1 \leq j \leq m-1$; 
\item $\overline{f}$ is also an o-polynomial; and  
\item $f(x+1)+f(1)$ is also an o-polynomial. 
\end{enumerate}
\end{theorem} 

\begin{theorem}\label{thm-basicproperty2}
Let $x^k$ be an o-polynomial on $\gf(2^m)$. Then every polynomial in 
$$\left\{x^{\frac{1}{k}},\, x^{1-k},\,  x^{\frac{1}{1-k}},\, x^{\frac{k}{k-1}},\, x^{\frac{k-1}{k}}\right\}$$ 
is also an o-polynomial, where $1/k$ denotes the multiplicative inverse of $k$ modulo $2^m-1$. 
\end{theorem} 

The following property of o-polynomials plays an important role in our construction of binary linear codes with 
o-polynomials. 

\begin{theorem}[\cite{CM11}]\label{thm-opoly2to1}
A polynomial $f$ from $\gf(2^m)$ to  $\gf(2^m)$ with $f(0)=0$ is an o-polynomial if and only if $f_u:=f(x)+ux$ 
is $2$-to-$1$ for every $u \in \gf(2^m)^*$. 
\end{theorem}

Let $f$ be any o-polynomial over $\gf(2^m)$. Define 
$$ 
f_u(x)=f(x)+ux 
$$
where $u \in \gf(2^m)^*$. It follows from Theorem \ref{thm-opoly2to1} that $f_u$ is 2-to-1 for every $u \in \gf(2^m)^*$. 
In the rest of Section \ref{sec-opolycodes}, we consider the codes $\C_{D(f_u)}$ defined by o-polynomials. By Theorem 
\ref{thm-opoly2to1}, the o-polynomial property of $f$ guarantees that the length of the code $\C_{D(f_u)}$ is equal to $2^{m-1}$ for any $u \in \gf(2^m)^*$. 
The dimension of $\C_{D(f_u)}$ usually equals $m$, but may be less than $m$. The minimum weight 
and the weight distribution of $\C_{D(f_u)}$ cannot be determined by the o-polynomial property alone, 
and differ from case to case.

\subsubsection{Binary codes from the translation o-polynomials} 

The translation o-polynomials are described in the following theorem \cite{Segre57}. 

\begin{theorem}
$\Trans(x)=x^{2^h}$ is an o-polynomial on $\gf(2^m)$, where $\gcd(h, m)=1$.
\end{theorem}

The following is a list of known properties of translation o-polynomials.  

\begin{enumerate}
\item $\Trans^{-1}(x)=x^{2^{m-h}}$ and  
\item $\overline{\Trans}(x)= xf(x^{2^m-2})=x^{2^m-2^{m-h}}$. 
\end{enumerate}

The proof of the following theorem is straightforward. 

\begin{theorem}\label{thm-translationcodes}
Let $f(x)=x^{2^h}$, where $\gcd(h,m)=1$. Then for any $u \in \gf(2^m)^*$, the code $\C_{D(f_u)}$ 
has parameters $[2^{m-1},\, m-1,\, 2^{m-2}]$ and is a one-weight code.  
\end{theorem}

The codes  $\C_{D(f_u)}$ in Theorem \ref{thm-translationcodes} have the same parameters as a subcode of 
the first order binary Reed-Muller code.

\subsubsection{Binary codes from the Segre and Glynn o-polynomials} 

The following theorem describes a class of o-polynomials, which are an extension of the original Segre 
o-polynomials. 

\begin{theorem}[\cite{DingYuan}]
Let $m$ be odd. Then 
$\Segre_a(x)=x^{6}+ax^4+a^2x^2$ is an o-polynomial on $\gf(2^m)$ for every $a \in \gf(2^m)$.
\end{theorem}

\begin{proof}
The conclusion follows from 
$$ 
\Segre_a(x)=(x+\sqrt{a})^6 + \sqrt{a}^3. 
$$ 
\end{proof}

We have the following remarks on this family of o-polynomials. 
\begin{enumerate}
\item $\Segre_0(x)=x^6$ is the original Segre o-polynomial \cite{Segre62,SegreBartocci}. So this is an extended family. 
\item $\Segre_a(x)=xD_5(x, a)=a^2D_5(x^{2^m-2}, a^{2^m-2})x^7$, where $D_5(x, a)=x^5+ax^3+a^2x$, 
          which is the Dickson polynomial of the first kind of order 5. 
\item $\overline{\Segre}_a=D_5(x^{2^m-2}, a)=a^2x^{2^m-2}+ax^{2^m-4}+x^{2^m-6}$. 
\item $\Segre_a^{-1}(x)=(x+\sqrt{a}^3)^{\frac{5\times 2^{m-1}-2}{3}} +\sqrt{a}$. 
\end{enumerate}

\begin{theorem}[\cite{DingYuan}]\label{thm-jan9}
Let $m$ be odd. Then 
\begin{eqnarray}\label{eqn-PayneInverse2}
\overline{\Segre}_1^{-1}(x)=\left( D_{\frac{3 \times 2^{2m}-2}{5}}(x, 1)\right)^{2^m-2}.  
\end{eqnarray}
\end{theorem}

Glynn discovered two families of o-polynomials \cite{Glynn83}. The first is described as follows. 

\begin{theorem}
Let $m$ be odd. Then $\Glynnone(x)=x^{3\times 2^{(m+1)/2}+4}$ is an o-polynomial.
\end{theorem}

\begin{conj} 
Let $m \ge 3$ be odd, and let $f(x)=x^{3\times 2^{(m+1)/2}+4}$ be the Glynn o-polynomial. 
When $m \in \{5, 7\}$, $\C_{D(f_u)}$ is a $[2^{m-1},\, m]$ code with the weight distribution 
of Table \ref{tab-semibentfcode6}. When $m \ge 9$, $\C_{D(f_u)}$ is a $[2^{m-1},\, m]$ 
code with five weights.  
\end{conj} 

An extension of the second family of o-polynomials discovered by Glynn is documented in the following theorem. 

\begin{theorem}[\cite{DingYuan}]
Let $m$ be odd. Then 
\begin{eqnarray*}
\Glynntwo_a(x)= \left\{ 
\begin{array}{ll}
x^{2^{(m+1)/2}+2^{(3m+1)/4}} + a x^{2^{(m+1)/2}} + (ax)^{2^{(3m+1)/4}}   & \mbox{ if } m \equiv 1 \pmod{4}, \\
x^{2^{(m+1)/2}+2^{(m+1)/4}} + a x^{2^{(m+1)/2}} + (ax)^{2^{(m+1)/4}}           & \mbox{ if } m \equiv 3 \pmod{4}.         
\end{array}
\right. 
\end{eqnarray*}
 is an o-polynomial for all $a \in \gf(q)$. 
\end{theorem}

\begin{proof}
Let $m \equiv 1 \pmod{4}$. Then 
$$ 
\Glynntwo_a(x)=(x+a^{(m-1)/4})^{2^{(m+1)/2}+2^{(3m+1)/4}} + a^{2^{(m+1)/2}+2^{(3m+1)/4}}. 
$$
The desired conclusion for the case $m \equiv 1 \pmod{4}$ can be similarly proved. 
\end{proof}

Note that $\Glynntwo_0(x)$ is the original Glynn o-polynomial. So this is an extended family. For some applications, 
the extended family may be useful.

For certain quadratic Boolean functions $f$, the code $\C_{D(f_u)}$ has good parameters and its weight distribution 
is known. The following result is an extended version of a result proved in \cite{Ding15}.

\begin{table}[ht]
\begin{center} 
\caption{The weight distribution of the codes of Theorem \ref{thm-hyperovalDS}}\label{tab-semibentfcode6}
\begin{tabular}{cc} \hline
Weight $w$ &  Multiplicity $A_w$  \\ \hline  
$0$          &  $1$ \\  
$2^{m-2}-2^{(m-3)/2}$ & $2^{m-2}+2^{(m-3)/2}$ \\  
$2^{m-2}$ & $2^{m-1}-1$ \\  
$2^{m-2}+2^{(m-3)/2}$ & $2^{m-2}-2^{(m-3)/2}$ \\ \hline 
\end{tabular}
\end{center} 
\end{table} 

\begin{theorem}\label{thm-hyperovalDS}
Let $\rho=2^i + 2^j$ and define 
$$ 
f_{u}(x)=x^\rho+ux, \, u \in \gf(2^m)^*. 
$$
If $f_u(x)$ is 2-to-1 on $\gf(2^m)$ and $\gcd(2^\kappa+1, \,2^m-1)=1$, where $\kappa=j-i$, 
then the binary code $\C_{D(f_u)}$ has parameters $[2^{m-1}, \,m, \,2^{m-2} - 2^{(m-3)/2}]$ and the weight 
distribution of Table \ref{tab-semibentfcode6} for any $u \in \gf(2^m)^*$.
\end{theorem}

The following $\rho$ satisfies the conditions of Theorem \ref{thm-hyperovalDS}: 
\begin{itemize} 
\item $\rho=6$ (Segre case). 
\item $\rho=2^\sigma + 2^\pi$ with $\sigma=(m+1)/2$ and $4\pi \equiv 1 \bmod{m}$ (Glynn I case).  
\end{itemize}

\begin{theorem}
Let $f(x)=xD_5(x, a)$, where $a \in \gf(2^m)$, and let $m$ be odd. Then the code  $\C_{D(f_u)}$ has parameters 
$[2^{m-1},\, m]$ and the weight distribution of Table \ref{tab-semibentfcode6} for any $u \in \gf(2^m)^*$.  
\end{theorem}

\begin{proof}
It is easily verified that $f(x)=(x+\sqrt{a})^6 + \sqrt{a}^3$. The desired conclusions then follow from Theorem 
\ref{thm-hyperovalDS} in the Segre case, i.e., $\rho=6$. 
\end{proof}

\begin{theorem}
Let $F(x)=x^{(5 \times 2^{m-1}-2)/3}$, and let $m$ be odd. Then code  $\C_{D(f_u)}$ has parameters 
$[2^{m-1},\, m]$ and the weight distribution of Table \ref{tab-semibentfcode6} for any $u \in \gf(2^m)$. 
\end{theorem} 

\begin{proof}
Note that the multiplicative inverse of $6$ modulo $2^m-1$ is equal to $(5 \times 2^{m-1}-2)/3$. The desired 
conclusion then follows from Theorem \ref{thm-hyperovalDS}. 
\end{proof}

\subsubsection{Binary codes from the Cherowitzo o-polynomials} 

The following describes another conjectured class of o-polynomials. 

\begin{conj}[\cite{DingYuan}]
Let $m$ be odd and $e=(m+1)/2$. Then 
$$
\Cherowitzo_a(x)=x^{2^e}+ax^{2^e+2}+a^{2^e +2}x^{3 \times 2^e +4} 
$$ 
is an o-polynomial on $\gf(2^m)$ for every $a \in \gf(2^m)$.
\end{conj}

We have the following remarks on this family. 
\begin{enumerate}
\item $\Cherowitzo_1(x)$ is the original Cherowitzo o-polynomial \cite{Ch88,Ch96}. So this is an extended family. 
\item No proof of the o-polynomial property is known in the literature.  
\item $\overline{\Cherowitzo}(x)=x^{2^m-2^e}+ax^{2^m-2^e-2}+a^{2^e +2}x^{2^m-3 \times 2^e -4}$. 
\item Carlet and Mesnager showed that $\Cherowitzo_1^{-1}(x)=x(x^{2^e+1}+x^3+x)^{2^{e-1}-1}$. 
\end{enumerate}

We can  prove the following. 

\begin{theorem}[\cite{DingYuan}]
$$ 
\Cherowitzo_a^{-1}(x)=x(ax^{2^e+1}+a^{2^e}x^3+x)^{2^{e-1}-1}. 
$$
\end{theorem} 

\begin{theorem}[\cite{DingYuan}]
$$ 
\overline{\Cherowitzo}=(ax^{2^m-2^e-2}+a^{2^e}x^{2^m-4}+x^{2^m-2})^{2^{e-1}-1}. 
$$
\end{theorem}

\begin{conj}
Let $m$ be odd, and let 
$$f(x)=b^{2^{(m+1)/2}+2}x^{2^{(m+1)/2}}+b^{2^{(m+1)/2}+1}x^{2^{(m+1)/2} +2}+ x^{3  \times 2^{(m+1)/2}+4},$$  
where $b \in \gf(2^m)$. 
If $m \in \{5, 7\}$, 
$\C_{D(f_u)}$ is a $[2^{m-1},\, m]$ code with at most five weights for every $u \in \gf(2^m)^*$.  
If $m \geq 9$, $\C_{D(f_u)}$ is a five-weight code with length 
$2^{m-1}$ and dimension $m$ for every $u \in \gf(2^m)^*$. 
\end{conj}

\subsubsection{Binary codes from the Payne o-polynomials} 

The following  documents a conjectured family of o-trinomials. 

\begin{conj}[\cite{DingYuan}]
Let $m$ be odd. Then 
$\Payne_a(x)=x^{\frac{5}{6}}+ax^{\frac{3}{6}}+a^2x^{\frac{1}{6}}$ is an o-polynomial on $\gf(2^m)$ for every $a \in \gf(2^m)$.
\end{conj}

We have the following remarks on this family. 
\begin{enumerate}
\item $\Payne_1(x)$ is the original Payne o-polynomial \cite{Pay85}. So this is an extended family. 
\item $\Payne_a(x)=xD_5(x^{\frac{1}{6}}, a)$.  
\item $\overline{\Payne}_a(x)=a^{2^m-3}\Payne_{a^{2^m-2}}(x)$. 
\item Note that 
$$ 
\frac{1}{6}=\frac{5 \times 2^{m-1}-2}{3}. 
$$
We have then 
$$ 
\Payne_a(x)=x^{\frac{2^{m-1}+2}{3}} + ax^{2^{m-1}} + a^2x^{\frac{5 \times 2^{m-1}-2}{3}}. 
$$
\end{enumerate}

\begin{theorem}[\cite{DingYuan}]
Let $m$ be odd. Then 
\begin{eqnarray}\label{eqn-PayneInverse}
\Payne_1^{-1}(x)=\left( D_{\frac{3 \times 2^{2m}-2}{5}}(x, 1)\right)^6 
\end{eqnarray}
and $\overline{\Payne}_1^{-1}(x)$ are an o-polynomial. 
\end{theorem}

\begin{proof}
Note that the multiplicative inverse of $5$ modulo is $\frac{3 \times 2^{2m}-2}{5}$. The conclusion then 
follows from the definition of the Payne polynomial and the fact that 
$$ 
D_5(x, 1)^{-1}=D_{\frac{3 \times 2^{2m}-2}{5}}(x, 1). 
$$ 
\end{proof}

\begin{conj}[\cite{DingYuan}]
Let $m$ be odd, and let 
$$f(x)=x^{\frac{5}{6}}+bx^{\frac{3}{6}}+b^2x^{\frac{1}{6}}=D_5\left(x^{\frac{5 \times 2^{m-1}-2}{3}}, \, b\right),$$  
where $b \in \gf(2^m)$. If $m \geq 7$, $\C_{D(f_u)}$ is a three-weight or five-weight code with length $2^{m-1}$ and dimension $m$ for 
all $u \in \gf(2^m)^*$.  
\end{conj}

\subsubsection{Binary codes from the Subiaco o-polynomials} 

The Subiaco o-polynomials are given in the following theorem \cite{Subiaco}. 

\begin{theorem}\label{thm-Subiaco}
Define 
$$ 
\Subiaco_a(x)=((a^2(x^4+x)+a^2(1+a+a^2)(x^3+x^2)) (x^4 + a^2 x^2+1)^{2^m-2}+x^{2^{m-1}},
$$ 
where $\tr(1/a)=1$ and $d \not\in \gf(4)$ if $m \equiv 2 \bmod{4}$. Then $\Subiaco_a(x)$ is an 
o-polynomial on $\gf(2^m)$. 
\end{theorem}

As a corollary of  Theorem \ref{thm-Subiaco}, we have the following. 

\begin{corollary} 
Let $m$ be odd. Then 
\begin{eqnarray}\label{cor-Subiaco}
\Subiaco_1(x)=(x+x^2+x^3+x^4) (x^4 + x^2+1)^{2^m-2}+x^{2^{m-1}} 
\end{eqnarray}
is an o-polynomial over $\gf(2^m)$. 
\end{corollary} 

Experimental data shows that the binary codes $\C_{D(f_u)}$ from the Subiaco o-polynomials have many weights 
and have smaller minimum weights compared with binary codes from other o-polynomials described in the previous 
subsections. Hence, the binary code $\C_{D(f_u)}$ from an o-polynomial could be very good and bad, depending on 
the specific o-polynomial $f$.

\subsection{Binary codes $\C_{D(f)}$ from functions on $\gf(2^m)$ of the form $f(x)=F(x)+F(x+1)+1$} 


A function $F(x)$ over $\gf(2^m)$ is called \emph{almost perfect nonlinear (APN)} if 
$$ 
\max_{a \in \gf(2^m)^*} \max_{b \in \gf(2^m)} \left|\{x \in \gf(2^m): F(x+a)-F(x)=b\}\right| =2. 
$$

Let $F$ be any function on $\gf(2^m)$. Define 
$$ 
f(x)=F(x)+F(x+1)+1. 
$$
For certain APN functions $F(x)$ over $\gf(2^m)$, it is known that $f$ is 2-to-1. 

\begin{conj}
Let $F(x)=x^{2^{(m-1)/2} + 3}$ and $m$ be odd. It is known that $F$ is both APN and AB. 
If $m \in \{5, 7\}$, $\C_{D(f)}$ is a three-weight code 
with length $2^{m-1}$ and dimension $m$. If $m \geq 9$, $\C_{D(f)}$ is a five-weight code with length 
$2^{m-1}$ and dimension $m$.  
\end{conj}

\begin{conj}
Let $F(x)=x^{2^{2h}-2^h+1}$, and let $\gcd(h, m)=1$. It is known that $F$ is both APN and AB.  

When $h=1$, $\C_{D(f)}$ is a $[2^{m-1}, m-1, 2^{m-2}]$ one-weight code.  

When $h \geq 2$ and $m$ is odd, $\C_{D(f)}$ is a three-weight or 
five-weight code with length $2^{m-1}$ and dimension $m$. 

In particular, when $h=3$ and $m$ is odd, $\C_{D(f)}$ is a three-weight code with length 
$2^{m-1}$ and dimension $m$ for every odd $m \geq 5$ and $m \not\equiv 0 \pmod{3}$. 
In this case, $d=57$ and the weight distribution of the code $\C_{D(f)}$ is given in Table 
\ref{tab-semibentfcode6}. 
\end{conj}

It is known that $f(x)=x^{2^{2h}-2^h+1} + (x+1)^{2^{2h}-2^h+1}$ is $2^s$-to-1, where $s=\gcd(h, m)$ \cite{Hertel}.

\begin{theorem}
Let $F(x)=x^{2^h+1}$, and let $\gcd(h, m)=1$. Then $\C_{D(f)}$ is a one-weight code with  
parameters $[2^{m-1},\, m-1,\, 2^{m-2}]$. 
\end{theorem}

\begin{proof}
The proof is straightforward and omitted, as  $f(x)$ is linear. 
\end{proof}

We have the following comments on other APN monomials. 
\begin{enumerate}[a)]
\item Let $F(x)=x^{2^m-2}$. Then $\C_{D(f)}$ is a binary code with length $2^{m-1}$ and 
dimension $m$, and has at most $m$ weights. The weights are determined by the Kloosterman 
sums.  
\item For the Niho function $F(x)=x^{2^{(m-1)/2}+2^{(m-1)/4}-1}$, where $m \equiv 1 \pmod{4}$, 
         the code  $\C_{D(f)}$ has 
         length $2^{m-1}$ and dimension $m$, but many weights. 
\item For the Niho function $F(x)=x^{2^{(m-1)/2}+2^{(3m-1)/4}-1}$, where $m \equiv 3 \pmod{4}$, 
         the code  $\C_{D(f)}$ has 
         length $2^{m-1}$ and dimension $m$, but many weights. 
\end{enumerate}
It would be extremely difficult to determine the weight distribution of the code $\C_{D(f)}$ for 
these three classes of APN monomials.

\subsection{Binary linear codes from some trinomials}

A lot of  constructions of cyclic difference sets in $(\gf(2^m)^*, \,\times)$ with the Singer parameters 
$(2^m-1,\, 2^{m-1},\, 2^{m-2})$ or  $(2^m-1,\, 2^{m-1}-1,\, 2^{m-2}-1)$ are proposed in the literature 
\cite{DD04,DingDScodes}. These difference sets can certainly be plugged into the second generic 
construction of this paper and obtain binary linear codes with good parameters. But determining the parameters 
of the binary linear codes may be difficult in general. 

There are also a number of conjectured cyclic difference sets in $(\gf(2^m)^*, \,\times)$. They give naturally 
binary linear codes with this construction. 
The following is a list of conjectured cyclic difference sets in $(\gf(2^m)^*, \,\times)$ with Singer parameters 
(see Chapter 4 of \cite{DingDScodes}). 

\begin{conj}\label{conj-DDSs} {\rm \cite[Chapter 4]{DingDScodes}}
For any $f \in \gf(2^m)[x]$, we define 
$$ 
D(f)^*=\{f(x): \,x \in \gf(2^m)\} \setminus \{0\}. 
$$  
Let $m \geq 5$ be odd. Then $D(f)^*$ is 
a difference set in $(\gf(2^m)^*, \,\times)$ with Singer parameters $(2^{m}-1, \,2^{m-1}, \,2^{m-2})$ for the following trinomials $f \in \gf(2^m)[x]${\rm :}  
\begin{enumerate}[a)]
\item $f(x)=x^{2^m-17} + x^{(2^m +19)/3} + x$. 
\item $f(x)=x^{2^m - 2^{m-4} -1} + x^{2^m  - (2^{m-2}+4)/3} + x$. 
\item $f(x)=x^{2^m-3} + x^{2^{(m+3)/2} + 2^{(m+1)/2} +4} + x$. 
\item $f(x)=x^{2^m - 2^{(m-1)/2} -1} + x^{2^{m-1} - 2^{(m-1)/2}} + x$. 
\item $f(x)=x^{2^m -2 - (2^{m-1} -2^2)/3} + x^{2^m -2^2 - (2^m -2^3)/3} + x$. 
\item $f(x)=x^{2^m - 2^{(m+1)/2} + 2^{(m-1)/2}} + x^{2^m - 2^{(m+1)/2} -1} + x$. 
\item $f(x)=x^{2^m-3(2^{(m+1)/2}-1)} + x^{2^{(m+1)/2} + 2^{(m-1)/2}-2} +x$. 
\item $f(x)=x^{2^m-2^{m-2}-1} + x^{2^{m-1}-2}+x$. 
\item $f(x)=x^{2^m-2^{(m+3)/2}-3} + x^{2^{(m+1)/2} +2} + x$. 
\item $f(x)=x^{2^m-3(2^{(m-1)/2}+1)} + x^{2^{m-1}-1} + x$. 
\item $f(x)=x^{2^m-5} + x^6 + x$. 
\end{enumerate}
\end{conj} 

For the linear codes $\C_{D(f)^*}$ of the conjectured difference sets $D(f)^*$ in Conjecture \ref{conj-DDSs}, we have the 
following conjectured parameters. 

\begin{conj}\label{conj-DDSscodes}
Let $m \geq 5$ and let $D(f)^*$ be defined as in Conjecture \ref{conj-DDSs}. Then for every $f$ given in Conjecture \ref{conj-DDSs}, 
the binary linear code $\C_{D(f)^*}$ has parameters $[2^{m-1}, \,m, \,2^{m-2}-2^{(m-3)/2}]$ and weight enumerator 
$$ 
1+(2^{m-2}-2^{(m-3)/2}) z^{2^{m-2}-2^{(m-3)/2}} + (2^{m-1}-1) z^{2^{m-2}} + (2^{m-2}+2^{(m-3)/2}) z^{2^{m-2}+2^{(m-3)/2}} .
$$
The dual code of $\C_{D(f)^*}$ has parameters $[2^{m-1},\, 2^{m-1}-m,\, 3]$. 
\end{conj}

Conjecture \ref{conj-DDSscodes} describes binary three-weight codes for the case that $m$ is odd. The next one is about 
binary three-weight codes for the case that $m$ is even. 

\begin{conj}\label{conj-DDSsJ170} 
Let $f(x)=x+x^{2^{(m-2)/2} + 2^{m-1}} + x^{2^{(m-2)/2} + 2^{m-1}+1} \in \gf(2^m)[x]$, where $m \equiv 2 \bmod{4}$ 
and $m \geq 6$. Define 
$$ 
D(f)=\{f(x): \,x \in \gf(2^m)\}. 
$$  
Then the binary code $\C_{D(f)}$ has parameters $[2^{m-1}, \,m, \,2^{m-2}-2^{(m-2)/2}]$ and 
weight enumerator 
$$ 
1+(2^{m-3}+2^{(m-4)/2})z^{2^{m-2}-2^{(m-2)/2}} + (3 \times 2^{m-2}-1) z^{2^{m-2}} + 
(2^{m-3}-2^{(m-4)/2})z^{2^{m-2}+2^{(m-2)/2}}.  
$$
It was conjectured in \cite[Chapter 4]{DingDScodes} that $D(f)^*$ is 
a difference set in $(\gf(2^m)^*, \,\times)$ with the parameters $(2^{m}-1, \,2^{m-1}-1, \,2^{m-2}-1)$. 
\end{conj}

\begin{table}[ht]
\begin{center} 
\caption{The weight distribution of the codes of Conjecture \ref{conj-DDSsJ171}}\label{tab-J171}
\begin{tabular}{cc} \hline
Weight $w$ &  Multiplicity $A_w$  \\ \hline  
$0$          &  $1$ \\  
$2^{m-2}-2^{(m-2)/2}$ & $2^{(m-2)/2}$ \\ 
$2^{m-2}-2^{(m-4)/2}$ & $2^{m-1}-2^{m/2}$ \\  
$2^{m-2}$ & $2^{m/2}+2^{(m-2)/2}-1$ \\  
$2^{m-2}+2^{(m-4)/2}$ & $2^{m-1}-2^{m/2}$ \\ \hline 
\end{tabular}
\end{center} 
\end{table}

Binary four-weight codes may also be produced with difference sets in $(\gf(2^m)^*, \,\times)$ as follows. 

\begin{conj}\label{conj-DDSsJ171} 
Let $f(x)=x+x^2+x^{2^m-2^{m/2} + 1} \in \gf(2^m)[x]$, where $m \geq 4$ and $m$ is even. Define 
$$ 
D(f)=\{f(x): \,x \in \gf(2^m)\}. 
$$  
Then the binary linear code $\C_{D(f)}$ has parameters $[2^{m-1}, \,m, \,2^{m-2}-2^{(m-2)/2}]$ and 
the weight distribution of Table \ref{tab-J171}. 
It was conjectured in \cite[Chapter 4]{DingDScodes} that $D(f)^*$ is 
a difference set in $(\gf(2^m)^*, \,\times)$ with parameters $(2^{m}-1, \,2^{m-1}-1, \,2^{m-2}-1)$. 
\end{conj}

The following is an another list of conjectured cyclic difference sets in $(\gf(2^m)^*, \,\times)$ with Singer parameters 
(see Chapter 4 of \cite{DingDScodes}). 

\begin{conj}\label{conj-DDSs2} {\rm \cite[Chapter 4]{DingDScodes}} 
For any $f \in \gf(2^m)[x]$, we define 
$$ 
D(f)=\{f(x(x+1)): x \in \gf(2^m)\}. 
$$
Let $m \geq 4$. Then $D(f)^*$ is 
a difference set in $(\gf(2^m)^*, \,\times)$ with Singer parameters $(2^{m}-1, \,2^{m-1}-1, \,2^{m-2}-1)$ for the following polynomials $f \in \gf(2^m)[x]${\rm :}   
\begin{enumerate}[1.]
\item $f(x)=x + x^{2^{(m+1)/2}-1} + x^{2^m-2^{(m+1)/2}+1} $, where $m$ is odd. 
\item $f(x)=x + x^{(2^m +1)/3} + x^{(2^{m+1} -1)/3}$, where $m$ is odd. 
\item $f(x)=x + x^{2^{(m+2)/2}-1} + x^{2^m-2^{m/2}+1}$, where $m$ is even.
\end{enumerate}
\end{conj} 

For the linear codes $\C_{D(f)^*}$ of the conjectured difference sets $D(f)^*$ in Conjecture \ref{conj-DDSs2}, we have the 
following conjectured parameters. 

\begin{conj}\label{conj-DDSscodes2}
Let $m \geq 4$ and let $D(f)$ be defined as in Conjecture \ref{conj-DDSs2}. Then for every $f$ given in Conjecture \ref{conj-DDSs2}, 
the binary linear code $\C_{D(f)^*}$ has parameters $[2^{m-1}-1, \,m-1, \, 2^{m-2}]$ and is a one-weight code. 
\end{conj} 

To determine the weight distribution of the code $\C_{D(f)}$ or $\C_{D(f)}^*$ of the conjectured difference sets 
listed in this section, one does not have to prove the difference set property of the set $D(f)$ or  $D(f)^*$.

\section{An expansion of the binary codes} 

\begin{table}[ht]
\begin{center} 
\caption{The weight distribution of the codes of Theorem \ref{thm-extebdedcode}}\label{tab-semibentfcode66}
\begin{tabular}{cc} \hline
Weight $w$ &  Multiplicity $A_w$  \\ \hline  
$0$          &  $1$ \\  
$2^{m-2}-2^{(m-3)/2}$ & $2^{m-1}+2^{(m-1)/2}$ \\  
$2^{m-2}$ & $2^{m}-2$ \\ 
$2^{m-2}+2^{(m-3)/2}$ & $2^{m-1}-2^{(m-1)/2}$ \\ 
$2^{m-1}$   & 1 \\ \hline
\end{tabular}
\end{center} 
\end{table} 

Let $\bone$ denote the all-one vector, that is, $(1,1,\ldots,1)$, of any length. The complement of any 
vector $\bc \in \gf(2)^n$ is defined to be $\bc+\bone$. For any binary code $\C$, we define 
$$ 
\overline{\C}=\C \cup \{\bc+\bone: \bc \in \C\}. 
$$ 
Then $\overline{\C}$ is a binary linear code, which has the same length as $\C$. For most of the binary codes 
$\C$ presented in this paper, the dimension of $\overline{\C}$ is one more than that of $\C$. In many cases, 
the weight distribution of $\overline{\C}$ can be deduced from that of $\C$. As an example, we have the 
following. 

\begin{theorem}\label{thm-extebdedcode}
Let $m$ be odd and let $\C$ be any binary linear code with parameters $[2^{m-1}, m]$ and the weight 
distribution of Table \ref{tab-semibentfcode6}. Then $\overline{\C}$ is binary linear code with parameters 
$[2^{m-1}, m+1]$ and the weight distribution of Table \ref{tab-semibentfcode66}. 
\end{theorem}  

\begin{example} 
When $m=5$, the code $\overline{\C}$ of Theorem \ref{thm-extebdedcode} has parameters $[16, 6, 6]$ 
and is optimal. 

When $m=7$, the code $\overline{\C}$ of Theorem \ref{thm-extebdedcode} has parameters $[64, 8, 28]$ 
and is almost optimal. 

\end{example}

\section{Concluding remarks} 

In this paper, we surveyed binary linear codes from Boolean functions and functions on $\gf(2^m)$ obtained 
from the second generic 
construction. Our focus was on such binary linear codes with at most five weights. Many one-weight codes, 
two-weight codes, three-weight codes, four-weight codes are presented in this paper. Some of them are 
optimal and some are almost optimal. The codes are also quite interesting in the sense that they may have 
applications in secret sharing \cite{ADHK,CDY05,YD06} and authentication codes \cite{CX05}. The 
parameters of some of the binary codes are different from those in \cite{CK85}, \cite{CG84}, \cite{CW84}, 
\cite{Choi}, \cite{FL07}, \cite{LiYueLi}, \cite{LiYueLi2},  \cite{Xia}, and \cite{ZD13}. 

A number of conjectures were presented in this paper as open problems. All the conjectures on difference sets, 
o-polynomials and the corresponding binary codes were confirmed for sufficiently many integers $m$ by Magma. 
The reader is warmly invited to attack these open problems. 

Finally, we make it clear that this is by no means a survey of all binary linear codes from Boolean functions, 
but a survey of binary linear codes from Boolean functions from the second generic construction described in 
Section \ref{sec-2ndgconst}.



\end{document}